\documentclass[twoside,twocolumn]{article}
\usepackage{acta}
\usepackage{amsthm,amsfonts,amsmath,amssymb,array}

\newcommand{\I}{\mathrm{i}}
\newcommand{\E}{\mathrm{e}}
\newcommand{\re}{\mathop{\mathrm{Re}}}
\newcommand{\im}{\mathop{\mathrm{Im}}}
\newtheorem{prop}{Proposition}

\newcommand{\Div}{\mathop{\mathrm{div}}}
\newcommand{\grad}{\mathop{\mathrm{grad}}}

\begin{document}
\Eingang{3}{9}{2015} \Annahme{3}{9}{2015} \Sachgebiet{1}
\Band{?} \Jahr{?}
\Heft{} \Ersteseite{1} \Letzteseite{}

\AuthorsI{M.~Yu.~Trofimov, S.~B.~Kozitskiy, A.~D.~Zakharenko}
\AddressI{Il'ichev Pacific Oceanological Institute,
43 Baltiyskay str., Vladivostok, 690041, Russia.\\
\hspace*{8pt}trofimov@poi.dvo.ru}

\Englishtitle{An energy flux conserving one-way coupled mode propagation model}
\Germanorfrenchtitle{}

\Kolumnentitel{A coupled mode propagation model}

\Englishabstract{%
A  pure analytic one-way coupled mode propagation model for resonant
interacting modes is obtained by the multiscale expansion method. It
is proved that the acoustic energy flux is conserved in this model
up to the first degree of the corresponding small parameter. The
test calculations with the COUPLE program give an excellent
agreement. } \Germanorfrenchabstract{}

\ScientificPaper

\section{Introduction}

The normal mode method is often used in the problems of acoustics.
It consists in the (local) separation of variables of the original
boundary value problem in such a way that on the cross-waveguide
direction we have a spectral problem, from which the normal modes
are obtained, and in the direction along the waveguide we have an
initial boundary value problem, determining the amplitudes of the
normal modes. In this formulation the field in a range-dependent
waveguide is expanded in terms of local modes with range-dependent
coefficients (mode amplitudes). In the adiabatic mode approximation
the independent propagation for each mode is assumed. For the
coupled mode propagation the derivation of the amplitude equations
is not that obvious, and it is considered here for the acoustic
case.

Adiabatic and coupled  mode acoustic equations appeared as a
convenient tool for solving problems of ocean acoustics since the
works of A. D. Pierce~\cite{pierce}, R. Burridge \& H.
Weinberg~\cite{Burridge}, J. A. Fawcett~\cite{fawcett}, and M. B.
Porter~\cite{couple,cracken}.
In all of these papers excepting~\cite{Burridge} the method of
multiscale expansions was not used~\cite{nay,tr_na}. In this work we
show that the systematic use of this method allows to obtain
unidirectional equations that produce, in the considered numerical
examples, the same results for the transmission losses as the COUPLE
2 way equations~\cite{cpl}, when the range-dependent waveguide is
approximated by range-independent stair steps and the coupled mode
solution is obtained by matching the solutions of the wave equation
for neighboring stair steps at their common boundary. Our approach
in principle can be extended to handle propagation in three
dimensions, whereas the discrete coupled mode method can not.
\par
For the obtained equations an important property of the acoustic energy
flux conservation is proved. Namely, if the difference between the wave numbers
of the modes is of order $1$ in the small parameter used in the
multiscale expansions method, then the acoustic flux (see the
definition and discussion in section~\ref{p_flux}) is of the same order, whereas the
acoustic Helmholtz equation possesses the property of the
energy flux conservation exactly.
\par
The boundary and interface conditions are derived simultaneously and
by the same method as the equations.

\section{Formulation of problem}\label{inter}
We consider the propagation of time-harmonic sound in the axially symmetric three-dimen\-si\-onal waveguide
$\Omega=\{(r,\phi,z)| 0 \leq r < \infty, 0\leq \phi < 2\pi, -H\leq z \leq 0 \}$ ($z$-axis is directed upward),
described by the acoustic Helmholtz equation
\begin{gather} \label{Helm}
\begin{aligned}
& \left(\gamma P_r \right)_r + \frac{1}{r}\gamma P_r
 + \left(\gamma P_z\right)_z + \gamma\kappa^2 P\\
& \qquad\qquad\qquad\qquad = \frac{-\gamma\delta(z-z_0)\delta(r)}{2\pi r}\,,
\end{aligned}
\end{gather}
where $\gamma = 1/\rho$, $\rho=\rho(r,z)$ is the density,
$\kappa(r,z)$ is the wavenumber. We assume the appropriate radiation
conditions at infinity in the $r,\phi$ plane, the pressure-release
boundary condition at $z=0$
\begin{gather} \label{Dir}
P=0\quad \text{at}\quad z=0\,,
\end{gather}
and the rigid boundary condition $\partial u/\partial z=0$ at $z=
-H$. The parameters of medium may be discontinuous at the
nonintersecting smooth interfaces  $z=h_1(r),\ldots,h_m(r)$, where
the usual interface conditions
\begin{gather}\label{InterfCond}
\begin{aligned}
& P_+ = P_-\,,\\
& \gamma_+ \left(\frac{\partial P}{\partial z}-h_r\frac{\partial P}{\partial r}\right)_+ =
\gamma_- \left(\frac{\partial P}{\partial z}-h_r\frac{\partial P}{\partial r}\right)_-
\end{aligned}
\end{gather}
are imposed. Hereafter we use the denotations $f(z_0,r)_+=\lim_{z\downarrow z_0}f(z,r)$ and
$f(z_0,r)_-=\lim_{z\uparrow z_0}f(z,r)$.
Without loss of generality we may consider the case $m=1$, so we set $m=1$ and denote $h_1$ by $h$.
\par
We introduce a small parameter $\epsilon$ (the ratio
of the typical wavelength to the typical size of medium inhomogeneities), the slow
variable $R=\epsilon r$ and postulate the following
expansions for the parameters $\kappa^2$, $\gamma$ and $h$:
\begin{gather*}
\kappa^2 = n_0^2(R,z) + \epsilon\nu(R,z)\,,\,\, \gamma  =
\gamma(R,z)\,,\,\, h = h(R)\,.
\end{gather*}
To model the attenuation effects we admit $\nu$ to be complex. Namely, we take $\im\nu = 2\eta\beta n_0$, where
$\eta = (40\pi\log_{10}e)^{-1}$ and $\beta$ is the attenuation  in decibels per wavelength. This implies that
$\im\nu\ge 0$.
\par
Consider a solution to the Helmholtz equation (\ref{Helm}) in the
form of the WKB-ansatz, where $\{\theta_j|j=M,\dots,N\}$ is a set of phases (fast variables):
\begin{gather}\label{ansatzR}
P = \sum_{j=M}^N(u_0^{(j)}(R,z) +\epsilon u_1^{(j)}(R,z) +\ldots)
\E^{\I\theta_j/\epsilon}\,.
\end{gather}
Introducing this anzatz into equation (\ref{Helm}), boundary condition (\ref{Dir}) and interface conditions
(\ref{InterfCond}), all rewritten in the slow variable, we obtain the sequence of the boundary value problems at each
order of $\epsilon$.

\section{The problem at $O(\epsilon^{0})$}
To obtain the normal modes we first consider the WKB anzats in the form $P=u_{0}^{(j)}(R,z)\E^{\I\theta_j(R,z)/\epsilon}$.
Further we can omit $j$.
From the equations at $O(\epsilon^{-2})$ and $O(\epsilon^{-1})$ we can conclude that $\theta$ is independent of $z$.

At $O(\epsilon^{0})$ now we have
\begin{gather}\label{E0}
(\gamma u_{0z})_z + \gamma n^2_0 - \gamma(\theta_R)^2u_0 = 0\,,
\end{gather}
with the interface conditions of the order $\epsilon^{0}$
\begin{gather}\label{InterfaceE0}
\begin{aligned}
&u_{0+} = u_0{-}\,,\\
&\left(\gamma \frac{\partial u_0}{\partial z}\right)_+ =
\left(\gamma \frac{\partial u_0}{\partial z}\right)_-\quad\mbox{at}\quad z=h\,,
\end{aligned}
\end{gather}
and boundary conditions $u=0$ at $z=0$ and $\partial u/\partial z=0$ at $z= -H$.
We seek a solution to problem (\ref{E0}), (\ref{InterfaceE0}) in the form
$u_0 = A(R)\phi(R,z)\,.$
From eqs.~(\ref{E0}) and (\ref{InterfaceE0}) we obtain the following spectral problem
for $\phi$ with the spectral parameter $k^2 = (\theta_R)^2$
\begin{gather} \label{Spectral}
\begin{aligned}
&\left(\gamma\phi_z\right)_z + \gamma n_0^2 \phi - \gamma k^2 \phi=0\,,\\
&\phi(0) = 0\,,\quad
 \frac{\partial \phi}{\partial z}=0\quad \text{at}\quad z= -H\,,\\
&\phi_+ = \phi_-\,,\\
&\left(\gamma \frac{\partial \phi}{\partial z}\right)_+ =
\left(\gamma \frac{\partial \phi}{\partial z}\right)_-\quad\mbox{at}\quad z=h\,.
\end{aligned}
\end{gather}
This spectral problem, being considered in the Hilbert space
$L_{2,\gamma_0}[-H,0]$ with the scalar product and the normalizing
condition
\begin{gather*} \label{L2scalar}
(\phi,\psi) = \int_{-H}^{\,0}\gamma \phi\psi\,dz\,,\quad
(\phi,\phi) = \int_{-H}^{\,0}\gamma \phi^2\,dz\,,
\end{gather*}
has countably many solutions $(k_j^2,\phi_j)$, $j=1,2,\ldots$ where the eigenfunction can be chosen as real functions.
The eigenvalues $k_j^2$ are real and have $-\infty$
as a single accumulation point \cite{nai}.

\section{The derivatives of eigenfunctions and wavenumbers with respect to $R$}
Before considering the problem at $O(\epsilon^{1})$ we should
consider the calculation of the derivatives of the eigenfunctions
and wavenumbers with respect to $R$.
\par
Differentiating spectral problem (\ref{Spectral}) with respect to $R$ we obtain the boundary value problem
for $\phi_{jR}$
\begin{gather}\label{phi_X}
\begin{aligned}
&\left(\gamma\phi_{jRz}\right)_z + \gamma n_0^2 \phi_{jR} - \gamma k_j^2 \phi_{jR} =
 - \left(\gamma_{R}\phi_{jz}\right)_z -\\
& \qquad\qquad (\gamma n_0^2)_R \phi_{j} + 2k_{jR}k_{j}\gamma \phi_{j}
+\gamma_{R} k_j^2 \phi_{j}\,, \\
& \qquad \phi_{jR}(0) = 0\,, \quad  \phi_{jRz}(-H) = 0\,,
\end{aligned}
\end{gather}
with interface conditions at $z=h$
\begin{gather}\label{InterfCond_x}
\begin{aligned}
&\phi_{jR+}-\phi_{jR-} = -h_{R}(\phi_{jz+}-\phi_{jz-})\,, \\
&\qquad\gamma_{+}\phi_{jRz+} - \gamma_{-}\phi_{jRz-} = -\left(\gamma_{R+}\phi_{jz+}-\right.\\
&\left.\gamma_{R-}\phi_{jz-}\right) -
 h_{R}\left(\left((\gamma\phi_{jz})_z\right)_+ -
\left((\gamma\phi_{jz})_z\right)_-\right)\,.
\end{aligned}
\end{gather}
We search a solution to problem (\ref{phi_X}), (\ref{InterfCond_x}) in the form
\begin{gather*}
\phi_{jR} = \sum_{l=0}^\infty C_{jl}\phi_l\,,
\quad\text{where}\quad
C_{jl} = \int_{-H}^0 \gamma \phi_{jR}\phi_{l}\,dz\,.
\end{gather*}
Multiplying (\ref{phi_X}) by $\phi_l$ and then integrating resulting
equation from $-H$ to $0$ by parts twice with the use of
interface conditions (\ref{InterfCond_x}), we obtain
\begin{gather*}\label{Cjl_2}
\begin{aligned}
& \left(k_l^2 - k_j^2\right) C_{jl}  = \int_{-H}^0 \gamma_{R}\phi_{jz}\phi_{lz} \,dz
+ 2 k_{jR}k_j\delta_{jl} - \\
& - \int_{-H}^0 \left(\gamma n_0^2\right)_R \phi_j\phi_l \,dz
 + k_j^2\int_{-H}^0 \gamma_{R} \phi_j\phi_l \,dz + \\
& \left\{ h_{R}(\gamma^2\phi_{jz}\phi_{lz})_+
\left[\left(\frac{1}{\gamma}\right)_+-\left(\frac{1}{\gamma}\right)_-\right]-
\right. \\
&  \left.
\left.
h_{R}\phi_j\phi_l\left[
                    \left(\gamma\left(k_j^2-n_0^2\right)\right)_+
                    - \left(\gamma\left(k_j^2-n_0^2 \right)\right)_-
                 \right]\vphantom{\frac{1}{\gamma}}\right\}\right|_{z=h},
\end{aligned}
\end{gather*}
where $\delta_{jl}$ is the Kronecker delta. Note that $(\gamma^2\phi_{jz}\phi_{lz})_+=(\gamma^2\phi_{jz}\phi_{lz})_-$.

\section{The problem at $O(\epsilon^{1})$}

We now represent a solution to the Helmholtz equation (\ref{Helm})
in the form of anzats (\ref{ansatzR}). At $O(\epsilon^{1})$ we
obtain
\begin{gather}\label{E1}
\begin{aligned}
&\sum_{j=M}^N\left(\left(\gamma u^{(j)}_{1z}\right)_z + \gamma n_0^2 u^{(j)}_{1} - \gamma k_j^2 u^{(j)}_{1}\right)
\E^{\I\theta_j/\epsilon} = \\
& \sum_{j=M}^N\left( -\mathrm{i}\gamma_{R}k_j u^{(j)}_0 -2\mathrm{i}\gamma k_j u^{(j)}_{0R} -\mathrm{i}\gamma k_{jR} u^{(j)}_0\right.\\
& \qquad\qquad\qquad\left. - \mathrm{i}\gamma k_j\frac{1}{R}u^{(j)}_0-\nu\gamma u^{(j)}_0\right)\E^{\I\theta_j/\epsilon}\,,
\end{aligned}
\end{gather}
with the boundary conditions $u^{(j)}_1=0$ at $z=0$, $\partial u^{(j)}_1/\partial z=0$ at $z=-H$,
and the interface conditions at $z=h(R)$:
\begin{gather}\label{InterfCondE1}
\begin{aligned}
& \sum_{j=M}^N (u^{(j)}_{1+}-u^{(j)}_{1-})\E^{\I\theta_j/\epsilon} = 0\,,\\
& \sum_{j=M}^N [\gamma_{+}u^{(j)}_{1z+}-\gamma_{-}u^{(j)}_{1z-}\\
&\qquad\qquad+\mathrm{i}k_j h_{R}u^{(j)}_0(\gamma_{-}-\gamma_{+})]\E^{\I\theta_j/\epsilon} = 0\,.
\end{aligned}
\end{gather}
We seek a solution to problem (\ref{E1}), (\ref{InterfCondE1}) in the form
\begin{gather*}\label{anz1}
\begin{aligned}
& u_1^{(j)} = \sum_{l=0}^\infty B_{jl}(X,Y)\phi_l(z,X)\,,\\
& \text{where}\quad B_{jl} = \int_{-H}^0 \gamma u_1^{(j)}\phi_l\, dz\,.
\end{aligned}
\end{gather*}
Multiplying (\ref{E1}) by $\phi_l$ and then integrating resulting
equation from $-H$ to $0$ by parts twice with the use of interface
conditions (\ref{InterfCondE1}), we obtain
\begin{gather*}\label{Bjl_1}
\begin{aligned}
&\sum_{j=M}^N \left(\left.(k_l^2-k_j^2)B_{jl}\right.\right.\\
&\left.\left.\qquad\qquad-A_j\mathrm{i} k_j h_{R}\phi_j\phi_l\left[\gamma_{+}-\gamma_{-}\right]
\right|_{z=h}\right)\E^{\I\theta_j/\epsilon} \\
& = \sum_{j=M}^N\left(-\mathrm{i}k_j A_j \int_{-H}^0 \gamma_{R}\phi_j\phi_l\,dz\right. \\
& -2\mathrm{i}k_j A_j \int_{-H}^0 \gamma\phi_{jR}\phi_l\,dz
 -2\mathrm{i}k_j A_{j,R} \int_{-H}^0 \gamma\phi_{j}\phi_l\,dz\\
& - \mathrm{i}k_{j,R} A_j \int_{-H}^0 \gamma\phi_{j}\phi_l\,dz
 - \mathrm{i}k_j\frac{1}{R}A_{j}\int_{-H}^0 \gamma\phi_{j}\phi_l\,dz\\
&\qquad\qquad\qquad\qquad\left.- A_j \int_{-H}^0 \nu\gamma\phi_{j}\phi_l\,dz\right)\E^{\mathrm{i}\theta_j/\epsilon}\,.
\end{aligned}
\end{gather*}
The terms $(k_l^2-k_j^2)B_{jl}$ in these expressions can be omitted
because of the resonant condition $|k_l-k_j|=O(\epsilon$).

As
\begin{gather*}
\begin{aligned}
& -\mathrm{i}k_j A_j \int_{-H}^0 \gamma_{R}\phi_j\phi_l\,dz
-2\mathrm{i}k_j A_j \int_{-H}^0 \gamma\phi_{jR}\phi_l\,dz  \\
& = \mathrm{i}k_j A_j\left(C_{lj} - C_{jl}\right)
-\mathrm{i}k_jA_jh_{R}\phi_j\phi_l\left.\left[\gamma_{+}-\gamma_{-}\right]\right|_{z=h}\,,
\end{aligned}
\end{gather*}
we get, after some algebra,
\begin{gather*}
\begin{aligned}
&  \sum_{j=M}^N\left( \vphantom{\int_{-H}^0
\nu\gamma\phi_{j}\phi_l\,dz}\mathrm{i}k_j A_j\left(C_{lj} -
C_{jl}\right)
 -2\mathrm{i}k_j A_{j,R} \delta_{jl} - \mathrm{i}k_{jR} A_j\delta_{jl}\right.\\
& \left.\qquad -\mathrm{i}k_j\frac{1}{R}A_{j} \delta_{jl}  - A_j
\int_{-H}^0 \nu\gamma\phi_{j}\phi_l\,dz
\right)\E^{\I\theta_j/\epsilon}=0\,.
\end{aligned}
\end{gather*}

\begin{prop}\label{pr_MPE}
The solvability condition for the problem at $O(\epsilon^1)$  is a
system of equations for $l=M,\ldots,N$
\begin{gather}\label{MPE}
2\mathrm{i}k_l A_{l,R} + \mathrm{i}k_{l,R} A_l +
\mathrm{i}k_l\frac{1}{R}A_{l} +
\sum_{j=M}^N\alpha_{lj}A_j\E^{\theta_{lj}}= 0\,,
\end{gather}
where $\alpha_{lj}$ and $\theta_{lj}$ are given by the following
formulas
\begin{gather}\label{alpha}
\begin{aligned}
& \alpha_{lj}  = \int_{-\infty}^0 \gamma\nu \phi_j\phi_l\,dz
 -\mathrm{i}k_j\left(C_{lj} - C_{jl}\right)\,,\\
& \theta_{lj}  = \frac{\I}{\epsilon}(\theta_j- \theta_l)\,.
\end{aligned}
\end{gather}
\end{prop}

\section{Energy flux conservation for equations~(\ref{MPE})}\label{p_flux}

The importance of the energy conservation law is not doubtful.
For the wave equation such a law can be established~\cite{lan}.
Since the Helmholtz equation is obtained from the wave equation, the energy conservation law
is converted to the acoustic energy flux conservation property.
This is the reason why it is widely accepted~\cite{couple} that energy flux conservation should be maintained in any propagation model.
The acoustic energy flux averaged over the period is defined as
$$
J(r,z)= \frac{1}{2\omega}\gamma\im((\grad P(r,z))P^*(r,z))\,.
$$
From now on we drop  the inessential factor $1/2\omega$.
As is well known, if $P$ is a solution of the Helmholtz equation~(\ref{Helm}) then the corresponding energy flux is conserved, that is
$
 \Div J(r,z)=0\,.
$
With our boundary conditions we have also the conservation property
$$
\Div \int_{-H}^{0}J(r,z)\,dz=0\,.
$$
\begin{prop}\label{pr_flux}
Assume that $\im\bar\nu=0$. Let $\{A_j|j=M,\ldots N\}$ be a solution to equations~(\ref{MPE}).
Then for
$\displaystyle{ P =  \sum_{j=M}^N A_j\phi_j\E^{\I\theta_j/\epsilon}}$
we have
$
\displaystyle{\Div\int_{-H}^{0}J(r,z)\,dz=O(\epsilon^2)}\,.
$
\end{prop}

\begin{proof}
First calculate the divergence in the general form for the anzats used:
\begin{gather}\label{eFlux}
\begin{aligned}
&\Div\int_{-H}^{0}J(r,z)\,dz = \frac{1}{r}\frac{\partial}{\partial
r}\left\{r\left[\sum_{l=M}^N
k_l|A_l|^2\right.\right.\\
& +\epsilon\sum_{l=M}^N\sum_{j=M}^N\im\left(C_{lj}A_lA_j^*\E^{\I(\theta_l-\theta_j)/\epsilon)}\right)\\
& \qquad\qquad\qquad\left.\left.+\epsilon\sum_{l=M}^N \im(A_{l,R}A_l^*)\right]\right\}\\
& =\epsilon\sum_{l=M}^N\sum_{j=M}^N(k_l-k_j)C_{lj}
\re\left(A_jA_l^*\E^{\I(\theta_j-\theta_l)/\epsilon}\right)\\
& +\epsilon\sum_{l=M}^N
(k_l|A_l|^2)_R+\epsilon\frac{1}{R}\sum_{l=M}^N
k_l|A_l|^2+O(\epsilon^2)\,.
\end{aligned}
\end{gather}
Consider now the sum on $l$ of equations (\ref{MPE}) multiplied
by $A_l^*$ minus the conjugate equations multiplied by $A_l$
\begin{equation*}
\begin{aligned}
\sum_{l=M}^N\left[\vphantom{\sum_{l=M}^N}\left(2\mathrm{i}k_l
A_{l,R} + \mathrm{i}k_{l,R} A_l +
\mathrm{i}k_l\frac{1}{R}A_{l}\right.\right. \\
+\sum_{j=M}^N\left.\alpha_{lj}A_j\E^{\theta_{lj}}\vphantom{\frac{1}{R}}
\right)A_l^* - \\
\left((-2\mathrm{i}k_l A^*_{l,R} - \mathrm{i}k_{l,R} A^*_l -
\mathrm{i}k_l\frac{1}{R}A^*_{l}\right. \\+ \left.
\sum_{j=M}^N\left.\alpha^*_{lj}A^*_j\E^{\theta^*_{lj}}\vphantom{\frac{1}{R}}\right)A_l\right]= 0\,.\\
\end{aligned}
\end{equation*}
After some transformation we have:
\begin{gather*}
\begin{aligned}
\sum_{l=M}^N\sum_{j=M}^N(\alpha_{lj}A_j\E^{\theta_{lj}}A_l^* - \alpha_{lj}^*A_j^*\E^{\theta_{lj}^*}A_l)\\
+\sum_{l=M}^N 2\mathrm{i}((k_l |A_{l}|^2 )_R +
k_l\frac{1}{R}|A_{l}|^2)=0\,,
\end{aligned}
\end{gather*}
then substitute for $\alpha_{lj}$ its expression (\ref{alpha})
\begin{gather*}
\begin{aligned}
\sum_{l=M}^N\sum_{j=M}^N \left(-\mathrm{i}k_j(C_{lj}-C_{jl})A_j\E^{\I(\theta_j-\theta_l)/\epsilon}A_l^* + \right. \\
\left. -\mathrm{i}k_j(C_{lj}-C_{jl}) A_j^*\E^{\I(\theta_l-\theta_j)/\epsilon}A_l \right)\\
+\sum_{l=M}^N 2\mathrm{i}((k_l |A_{l}|^2 )_R +
k_l\frac{1}{R}|A_{l}|^2) =0\,,
\end{aligned}
\end{gather*}
and collect terms
\begin{gather*}
\begin{aligned}
\sum_{l=M}^N\sum_{j=M}^N \left(-\mathrm{i}k_j(C_{lj}-C_{jl})2\re(A_j\E^{\I(\theta_j-\theta_l)/\epsilon}A_l^*)\right)\\
+\sum_{l=M}^N 2\mathrm{i}((k_l |A_{l}|^2 )_R +
k_l\frac{1}{R}|A_{l}|^2) = 0\,,
\end{aligned}
\end{gather*}
write double sums separately for terms with $C_{lj}$ and $C_{jl}$
\begin{gather*}
\begin{aligned}
\sum_{l=M}^N\sum_{j=M}^N \left(-\mathrm{i}k_jC_{lj}2\re(A_j\E^{\I(\theta_j-\theta_l)/\epsilon}A_l^*)\right)\\
+\sum_{l=M}^N\sum_{j=M}^N \left(\mathrm{i}k_jC_{jl}2\re(A_j\E^{\I(\theta_j-\theta_l)/\epsilon}A_l^*)\right)\\
+\sum_{l=M}^N 2\mathrm{i}((k_l |A_{l}|^2 )_R +
k_l\frac{1}{R}|A_{l}|^2) = 0\,,
\end{aligned}
\end{gather*}
exchange indexes $l$ and $j$ in the second double sum and finally
get
\begin{gather*}
\begin{aligned}
\sum_{l=M}^N\sum_{j=M}^N \left(\mathrm{i}(k_l-k_j)C_{lj}2\re(A_j\E^{\I(\theta_j-\theta_l)/\epsilon}A_l^*)\right)\\
+\sum_{l=M}^N 2\mathrm{i}((k_l |A_{l}|^2 )_R +
k_l\frac{1}{R}|A_{l}|^2) = 0\,.
\end{aligned}
\end{gather*}
The last equation coincides modulo $2\mathrm{i}$ with the
$O(\epsilon)$-part of (\ref{eFlux}).
\end{proof}

\begin{figure*}
\includegraphics[width=\textwidth]{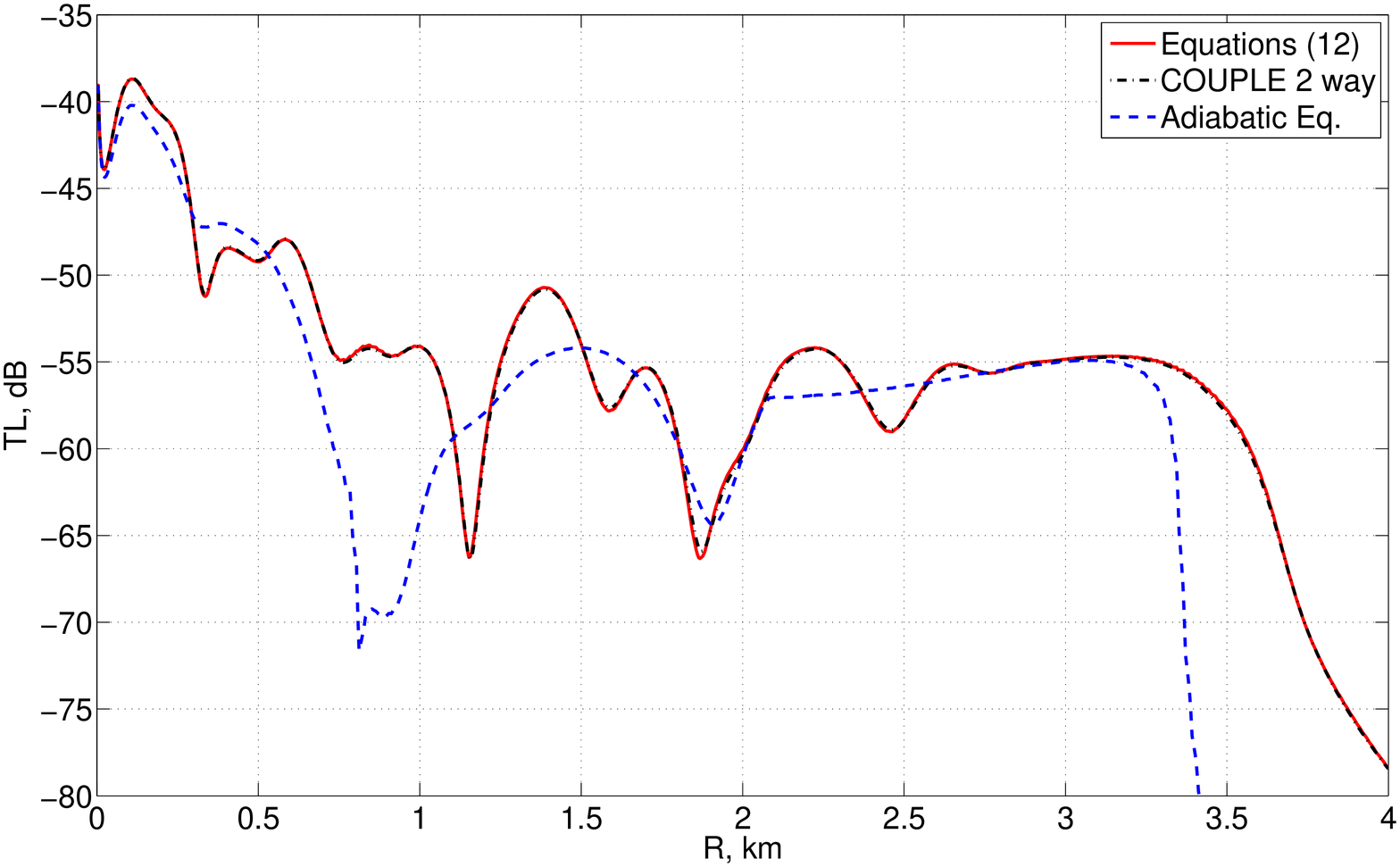}
\caption{The result of the numerical simulation with coupled-mode model
(\ref{MPE}) for the standard ASA wedge benchmark with absorbing
bottom compared to solutions obtained by the COUPLE
program~\cite{cpl} and the adiabatic mode equation. The source is
placed at 100 m, and the depth of the receiver is 30 m.}\label{Fig1}
\end{figure*}

\begin{figure*}
\includegraphics[width=\textwidth]{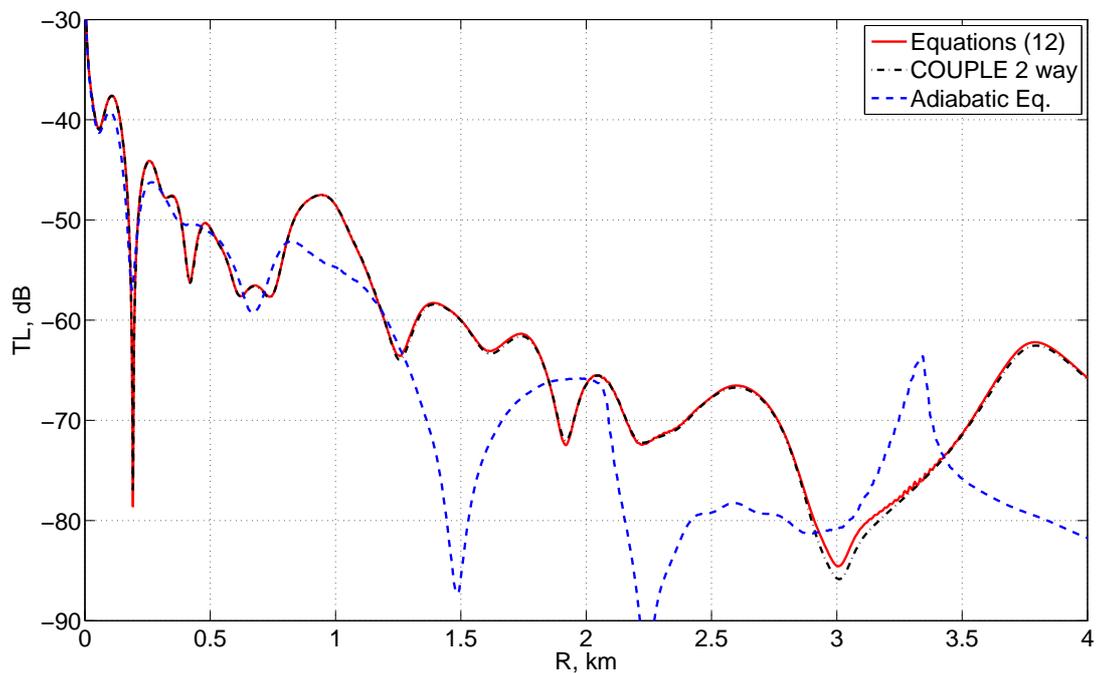}
\caption{The result of the numerical simulation with coupled-mode model
(\ref{MPE}) for the standard ASA wedge benchmark with absorbing
bottom compared to solutions obtained by the COUPLE
program~\cite{cpl} and the adiabatic mode equation. The source is
placed at 100~m, and the depth of the receiver is 150~m.}\label{Fig2}
\end{figure*}

\section{Numerical examples}

For the standard ASA wedge benchmark with the angle of wedge $\approx
2.86^\circ$  we numerically simulate the sound propagation
to illustrate the efficiency of our coupled-mode model.
The bottom depth decreases linearly from $200\,m$ at $X=0$ to
zero at $X=4\,km$.
The sound speed in the water is $1500\,m/sec$.
The sound speed in the bottom, which is considered liquid, is $1700\,m/sec$.
The bottom density is $1500\,kg/m^3$, the water density is $1000\,kg/m^3$.
We assume that there is no attenuation in the water layer, and in the bottom the attenuation
is $0.5\,dB/\lambda$. For the calculation purposes we restrict the total depth by
$1500\,m$ and suppose that in the bottom the absorption increases linearly from $0.5\,dB/\lambda$
at depth $1000\,m$ to $2.5\,dB/\lambda$ at depth $1500\,m$.
\par
The point source of frequency 25~Hz is placed at the depth of 100~m. In this case we have $44$ propagating modes.
Figure~\ref{Fig1} illustrates the transmission loss for the receiver
depth of 30~m.
Comparison with the numerical solution obtained by the COUPLE
program~\cite{cpl} shows that the mean square
difference between the two curves is about 0.15~dB.
A similar result is presented in figure \ref{Fig2} where the transmission
loss for the ASA wedge benchmark at depth 150 m is depicted.
In this case the mean square difference between our curve and the curve produced
with the COUPLE 2 way program is about 0.4~dB.

\section{Conclusion}
In this article a one-way coupled mode propagation model for the
resonantly interacting modes has been introduced.
The acoustic energy flux is conserved for the model with the accuracy
adequate to the used approximation (proposition~\ref{pr_flux}).
The test calculations were done for the ASA wedge benchmark and
proved excellent agreement with the COUPLE program~\cite{cpl}.

\section*{Acknowledgements}
The authors are grateful for the support of ``Exxon Neftegas Limited'' company.


\end{document}